\documentclass[journal abbreviation]{copernicus}

\usepackage{graphicx}
\usepackage{dcolumn}
\usepackage{bm}
\usepackage{amsmath}
\usepackage{graphicx}
\usepackage{caption}
\usepackage{amsmath}
\usepackage{enumitem}
\usepackage{amsthm}
\usepackage{xcolor}
\usepackage{mathptmx}
\usepackage[page,toc,titletoc,title]{appendix}
\usepackage{comment}
\usepackage[overload]{empheq}
\usepackage{cases}
\usepackage{lipsum}
\usepackage{amsfonts}
\usepackage{hyperref}
\usepackage{float}
\usepackage{natbib}
\usepackage{placeins}
\usepackage{atbegshi}
\newcommand{\dd}{\mathrm{d}}

\newcommand{\SSS}{\mathrm{S}}

\newcommand{\RRR}{\mathrm{R}}

\newtheorem{proposition}{Proposition}
\usepackage{hyperref}

\definecolor{darkorange}{rgb}{1.0, 0.55, 0.0}

\definecolor{rred}{rgb}{0.7,0,0.1}
\definecolor{greenrb}{rgb}{0.2,0.6,0.2}

\def\bea{\begin{equation} \begin{aligned}}
\def\eea{\end{aligned} \end{equation}}
\def\be{\begin{equation}}
\def\ee{\end{equation}}

\begin{document}

\nolinenumbers
\title{Dynamical regimes of CCN activation in adiabatic air parcels}

\Author[1][manuel.santos-gutierrez@weizmann.ac.il]{Manuel}{Santos Guti\'errez} 
\Author[1]{Micka\"el D.}{Chekroun}
\Author[1]{Ilan}{Koren}

\affil[1]{Department of Earth and Planetary Sciences, Weizmann Institute of Science, Rehovot 76100, Israel}

\runningauthor{MSG, MDC, and IK}
\runningtitle{Unveiling the Activation Puzzle in Elusive Twilight Clouds}


\received{}
\pubdiscuss{} 
\revised{}
\accepted{}
\published{}


\firstpage{1}

\maketitle

\begin{abstract}
Ubiquitous, yet elusive to a complete understanding: Tiny, warm clouds with faint visual signatures play a critical role in Earth's energy balance.  These "twilight clouds", as they are sometimes called, form under weak updraft conditions. Their constituent particles exist in a precarious state, teetering between hazy wisps and activated droplets. This delicate thermodynamic balance creates a limited reservoir of supersaturation, which activated droplets readily consume. Our research presents a novel approach, solving coupled equations for particle growth (Köhler's equation) and supersaturation change. This reveals previously unconsidered activation states for these clouds. Additionally, the analysis predicts conditions where particles can exhibit self-sustained oscillations between haze and activated droplet states.
\end{abstract}



\section{Introduction}

The coupling of cloud microphysics with climate dynamics is one of the most challenging problems in the atmospheric sciences. Far from being solved, the problem of cloud parameterization is responsible for the largest sources of uncertainty in global climate projections \citep{zelinka2020}. The large spatial scale separation between general circulation and convective motions, the inherent complexity of cloud motion and the current computational power, demand a more accurate description of the nonlinear processes in the onset and development of clouds.

From all the cloud types, the warm and smallest ones are particularly difficult to parameterize in climate models despite being vastly abounding in Earth's atmosphere \citep{koren2007twilight,varnai2009modis,eytan2020}. Specifically, "twilight clouds" are the most challenging to capture due to their poor optical signature. Regarding their microphysical structure, these small clouds are likely to reside near the transition from haze to activated cloud droplets, a domain of sizes difficult to resolve in current large-scale numerical models. Some of the latter clouds were shown to form in buoyant humidity pockets--- often forming subLCL clouds \citep{hirsch2015properties,hirsch2017,altaratz2021environmental}---, that oscillate between these two thermodynamic states due to the weak updraft in which they are embedded. The transition from haze to activated cloud droplets is described by K\"ohler's theory. This study aims to understand the dynamics of droplet's transitions relevant to twilight clouds.

While Köhler's theory \citep{kohler_1936} explains how individual cloud condensation nuclei (CCN) activate into droplets based on critical humidity and radius, it treats single particles. For a real cloud with many CCNs in the same space, we need to consider how growing droplets affect the surrounding supersaturation \citep{Korolev2003}, \cite[Eq.~42]{pinsky2013}. Cloud models typically couple droplet growth with supersaturation change after activation, e.g. \citep{khain2008,khain_pinsky_2018}. Köhler's theory assumes slow supersaturation depletion compared to droplet growth, which works well for vigorous clouds. In contrast, this study focuses on the fate of collections of monodisperse particles taken together in situations with weak updrafts, where supersaturation sources are limited and droplet sizes remain near critical values.

\section{The supersaturation-radius-K\"ohler (SRK) equation}\label{sec:theory}

We consider a vertically moving adiabatic air parcel seeded with a monodisperse family of aerosols of the same size and chemistry. The small size of CCN provokes that their curvatures increase the equilibrium vapor pressure, so that humidity values above saturation are needed to activate a water droplet. The curvature effect is also counteracted by the CCN's chemical ability to retain water and their combination is the content of K\"ohler theory \citep{kohler_1936,prupp}. 

If $\SSS$ denotes liquid-water supersaturation and $r$ the radius of a particle, the rate of change of $r$, with respect to time, is assumed to take the form:
\begin{equation}\label{eq:kohler}
	\dot{r} = \frac{D}{r}\left( \SSS - \frac{A}{r} + \frac{B}{r^3}  \right),
\end{equation}
where $A/r$, is Kelvin's curvature effect and, $B/r^3$ is Raoult's term and $D$ is the temperature-dependent diffusional parameter; see Appendix~\ref{app:notation}. The right-hand-side of Eq.~\eqref{eq:kohler} is a truncation of the K\"ohler curve which originally takes an exponential form, see e.g. \cite{wallace_hobs,arabas2017,prabhakaran2020}. For large values of $r$, saturation, $\SSS > 0$ is enough to induce condensational growth, since $\dot{r}>0$.

Analyzing the function $f(r) = Ar^{-1}-Br^{-3}$ reveals a single global maximum at $r_c = \sqrt{3B/A}$ and associated supersaturation $\SSS_c = f(r_c) = Ar_c^{-1}-Br_c^{-3}=2A/3r_c$. Such values are called the critical radius and supersaturation, respectively. Indeed, if $\SSS > \SSS_c$, it follows that equation \eqref{eq:kohler} does not have an equilibrium and therefore, the particle radius $r(t)$ grows indefinitely, entering the activated state. If $0<\SSS<\SSS_c$, two equilibria are allowed, of radii $r_h<r_c$ and $r_u>r_c$. The equilibrium $r_h$ is stable and represents haze, i.e., humidified particles which do not grow with the available supersaturation. The equilibrium $r_u$ is unstable. It reflects a threshold size beyond which a particle activates and becomes hygroscopic. An example of K\"ohler curve, $f(r)$, for salt is shown in Fig.~\ref{fig:illus_kohler}. As supersaturation increases, the haze and unstable equilibrium collide, provoking  model \eqref{eq:kohler} to undergo a saddle-node bifurcation leading to the activation of a droplet \citep{arabas2017}.
\begin{figure}[h]
	\centering
	\includegraphics[width=.8\linewidth, height=.3\textwidth]{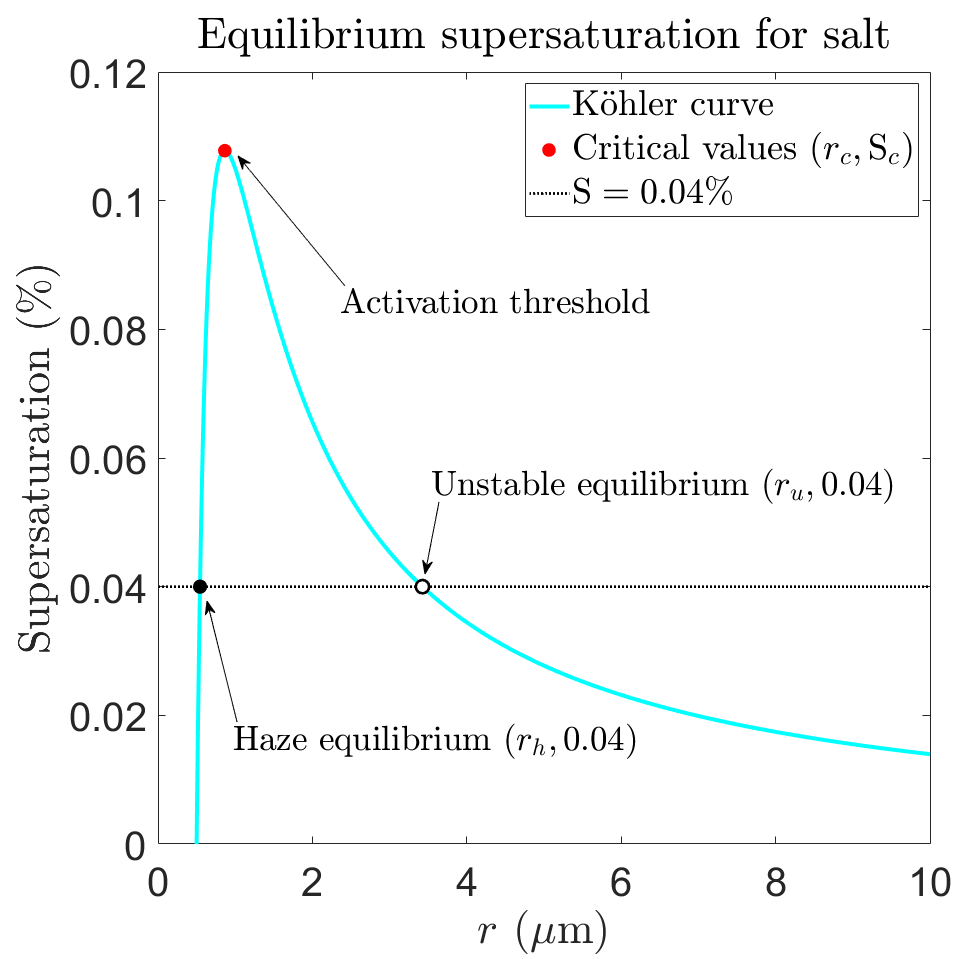}
	\caption{\label{fig:illus_kohler}The cyan curve is the K\"ohler curve for the equilibrium supersaturation for a small salt aerosol. Here, we consider a salt particle of dry radius $0.065\mu\mathrm{m}$. The dashed line indicates an example fixed supersaturation equation to $0.04\%$. The surface tension parameter is $A=1.4\cdot 10^{-3}\mu \mathrm{m}$ and the chemistry coefficient, $B = 3.5\cdot 10^{-4}\mu \mathrm{m}^3$.}
\end{figure}

 K\"ohler's critical values are often used as an activation threshold (red dot in Fig.~\ref{fig:illus_kohler}). 
However, they assume a slowly depleting humidity reservoir, limiting their applicability to describing the collective activation of many droplets. When supersaturation (S) becomes time-dependent, its sources and sinks become crucial. By taking the time derivative of S, we arrive at its evolution equation \citep{squires_1952,Korolev2003}:
\begin{equation}\label{eq:s_mixing_ratio}
	\frac{1}{1+\SSS}\frac{\dd \SSS}{\dd t} = \tau^{-1} - \beta\frac{\dd q}{\dd t},
\end{equation}
Here, $\beta$ (defined in Appendix \ref{app:notation}) accounts for the latent heat released during condensation. The variable $q$ represents the liquid-water mixing ratio, and $\tau^{-1}$ is the supersaturation timescale. This timescale is typically determined by the product of vertical velocity ($w$) and the adiabatic parameter ($a_0$), as detailed in Appendix \ref{app:notation}; see also \citep{prupp}. In this specific case,  $\tau^{-1} = a_0w$ reflects the increase in supersaturation due to adiabatic cooling during ascent.

In the case of monodisperse droplets, $r$ represents the average radius of the droplets.
In Equation \eqref{eq:s_mixing_ratio}, neglecting higher-order terms in parameters $A$ and $B$, the condensation rate ($\dd q/ \dd t $) becomes proportional to the product of the mean radius ($r$) and supersaturation (S), expressed as:
\begin{equation}\label{eq:mixing_approximation}
	\beta \frac{\dd q}{\dd t} \simeq \alpha (N) r\SSS.
\end{equation}
This proportionality is thus captured by the coefficient $\alpha(N)$, defined as:
\begin{equation}\label{eq:alpha2}
	\alpha(N)  = 4\pi \rho_w \beta D N r.
\end{equation}
Here, $\rho_w$ is the water density, $N$ is the number concentration of droplets (particles per unit volume), $\beta$ is a constant related to latent heat release (see Appendix \ref{app:notation}), and $D$ is a diffusion coefficient. The dependence of $\alpha (N)$ on particle concentration ($N$) is often omitted when $N$ is not a key factor.

The term $rS$ in Equation \eqref{eq:mixing_approximation} resembles a predator-prey term. Larger particles and higher supersaturation values ($rS$) correspond to a higher rate of humidity consumption. Conversely, smaller particles and lower supersaturation lead to less humidity consumption.

Furthermore, $\alpha (N)$  is related to the phase-relaxation timescale ($\tau_p$) through the relation $\tau_p = (r\alpha(N))^{-1}$. This timescale represents the time required for supersaturation to return to equilibrium  \citep{Korolev2003,prabhakaran2020}.

The supersaturation-radius-K\"ohler (SRK) equations,
 with Eqns~\eqref{eq:s_mixing_ratio}, \eqref{eq:mixing_approximation}, and \eqref{eq:kohler} corresponding to each term, describe CCN activation within a cloud parcel that has external sources of supersaturation. 
In warm clouds, supersaturation typically remains below a few percent \citep{khain_pinsky_2018,altaratz2021environmental}.
Therefore,
 it is common practice to simplify Eq.~\eqref{eq:s_mixing_ratio} using the approximation  $1+\SSS \approx 1$ (see, e.g., \cite{khain_pinsky_2018}). However,
 this simplification neglects the nonlinear behavior of supersaturation, which can become significant at higher values and is crucial for understanding the distribution of supersaturation in turbulent clouds \citep{santosgutierrez_2024}. For the case of monodisperse cloud particles embedded in an adiabatic air volume,
 the SRK equations governing their growth rate can be expressed as:
\begin{subequations}\label{eq:cond_growth} 
	\begin{empheq}[left = \empheqlbrace]{align}
		\dot{\SSS} &= \tau^{-1} - \alpha(N) r \SSS \label{eq:srka} \tag{SRKa} \\
		\dot{r^2} &= 2D\left( \SSS - f(r) \right), \label{eq:srkb} \tag{SRKb}
	\end{empheq}
\end{subequations}
where $\SSS$ denotes supersaturation, and $f$, the K\"ohler curve.

\subsection{Stability analysis for activation}

In line with K\"ohler's equation~\eqref{eq:kohler}, the activation/deactivation of a cloud droplet corresponds to the instability/stability of the variable $r^2$ of the SRK equation. We thus perform the linear stability analysis. Setting $\dot{\SSS} = \dot{r^2}=0$, we find a unique equilibrium $(\SSS_0,r^2_0)$, satisfying $r^2_0>0$:
\begin{subequations}\label{eq:equilibrium}
	\begin{align}
		\SSS_0 &= \frac{1}{\tau \alpha}\sqrt{\frac{A-\frac{1}{\tau \alpha}}{B}},\\
		r^2_0 &= \frac{B}{A-\frac{1}{\tau \alpha}},
	\end{align}
\end{subequations}
where $(\tau\alpha)^{-1}$ has the units of length. The immediate observation is that the equilibrium $(\SSS_0,r^2_0)$ must lie on the K\"ohler curve in the $\SSS$-$r^2$-space,  as shown in \eqref{eq:srkb}.  Moreover, the value of $\tau^{-1}$, as the source term for $\SSS$, determines the following: 
\begin{enumerate}
	\item[(i)] If  $\tau^{-1} = 0$, then $r^2_0 = r^2_c /3 = A/B$.
	\item[(ii)] If $\tau^{-1}<A\alpha$, $(\SSS_0,r^2_0)$ is an equilibrium.
	\item[(iii)] If $\tau^{-1}=2A\alpha/3$, $r^2_0 = r^2_c$.
	\item[(iv)] As $\tau^{-1}$ approaches $A\alpha$, we have $\lim_{\tau^{-1} \rightarrow A\alpha^+}r^2_0 = \infty$.
	\item[(v)] If $\tau^{-1} > A\alpha$, there are no equilibria.
	\item[(vi)] If supersaturation is decreased, $\tau^{-1}<0$, $\SSS_0<0$ and $r^2_0>0$.
\end{enumerate}
The presence of equilibria is determined by the interplay of the $(\tau\alpha)^{-1}$-factor and the curvature coefficient, $A$, whereas the chemistry coefficient $B$ only determines the location of the equilibrium.

\begin{figure}
	\centering
	\includegraphics[width=1\linewidth, height=.25\textwidth]{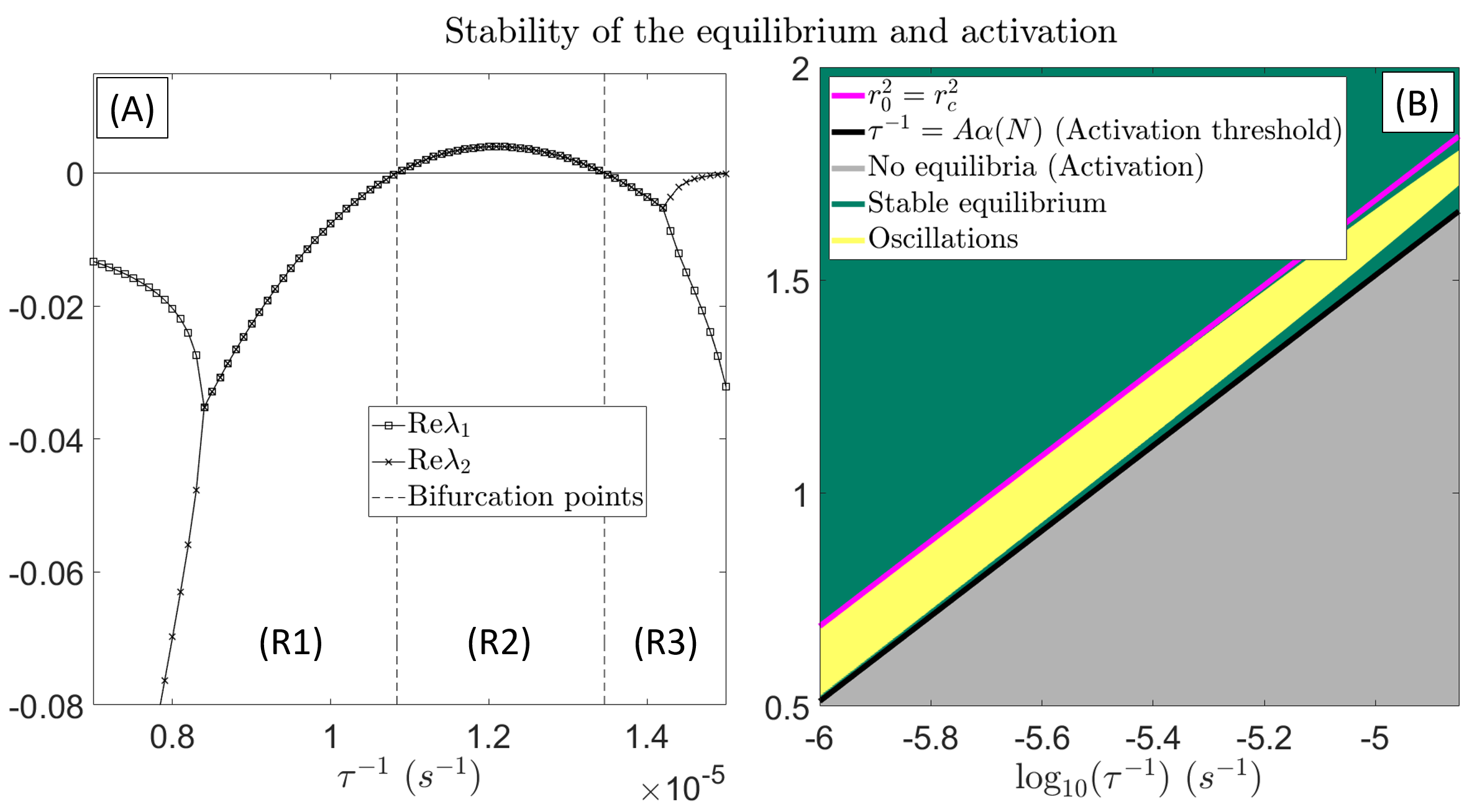}
	\caption{\label{fig:landscape_alpha1} \textbf{Panel (A):} The squared and crossed lines show the real parts of the eigenvalues $\lambda_1,\lambda_2$ of the Jacobian of the SRK equation evaluated at the equilibrium $(\SSS_0,r^2_0)$ of Eq.~\eqref{eq:equilibrium} as a function for $\tau^{-1}$, for $N=50\mathrm{cm}^{-3}$. The vertical dashed lines indicate where the equilibrium loses stability. \textbf{Panel (B):} The dark-green path corresponds to configurations in $\tau^{-1}$ and $N$ where the SRK equation is linearly stable. At the gray area, the activation threshold is surpassed and the equilibrium Eq.~\eqref{eq:equilibrium} is not well defined. At the yellow area, the SRK system is linearly unstable and supports oscillations. The magenta and black lines satisfy $r_0^2 = r_c^2$ and $\tau^{-1} = A\alpha(N)$, respectively.}
\end{figure}

The Jacobian matrix of the SRK equation determines the linear stability of the equilibrium $(r^2_0,\SSS_0)$. A careful analysis of such matrix reveals the following results--- the calculations are detailed in appendices~\ref{ap:proof}~and~\ref{app:proof_2}---:
\begin{enumerate}
    \item[(i)] If $\tau^{-1}<2A\alpha/3$, $(r^2_0,\SSS_0)$ is stable and therefore, CCN do not activate.
    \item[(ii)] If $\alpha < \frac{4A^3}{243B^2}D$, $(r^2_0,\SSS_0)$ is linearly unstable for some $2A\alpha/3<\tau^{-1}<A\alpha$.
    \item[(iii)] If $\alpha > \frac{4A^3}{243B^2}D$ and $2A\alpha/3<\tau^{-1}<A\alpha$, $(r^2_0,\SSS_0)$ is stable, although $r^2_0>r_c^2$.
\end{enumerate}
The proof of item (i) is found in Appendix~\ref{ap:proof}, and that of items (ii) and (iii) in Appendix~\ref{app:proof_2}.

It follows from the activation threshold $\alpha(N)A$, that the larger particle concentration $N$ is, the more difficult it is to activate cloud droplets. Contrarily, if $\alpha(N)\approx 0$, a minimal updraft will make $\tau^{-1}$ positive and induce activation. In case of $f(r)$ not being present in Eq.~\eqref{eq:srkb}, the system does not admit equilibria. As can be observed by direct computation; see \cite{Devenish2016} for a more detailed analytical work.

\section{Case study}\label{sec:case_study}

We evaluate our analytical framework for an air parcel at a temperature and pressure of $T=283\mathrm{K}$ and $P = 10^5\mathrm{Pa}$. We assume that the cloud is seeded with salt CCN with dry radii of $0.065\mu\mathrm{m}$. This gives the curvature and chemistry coefficients of $A=1.4\cdot 10^{-3}\mu\mathrm{m}$ and $B = 3.5\cdot 10^{-4}\mu\mathrm{m}^3$. Thus, the critical K\"ohler parameters are: 
\begin{subequations}
	\begin{align}
		r^2_c & = 0.75 \mu \mathrm{m}^2,\\
		\SSS_c &= 0.108 \%.
	\end{align}
\end{subequations}
The rest of model parameters are found in Table~\ref{tab:1}. 
\begin{center}
	\begin{table}
		\centering
		\caption{\label{tab:1}Parameter configuration for the SRK equation. $T = 283\mathrm{K}$ and $P = 10^5 \mathrm{Pa}$.} 
	\renewcommand{\arraystretch}{1.8}
	\setlength{\tabcolsep}{4pt}
	\begin{tabular}{ccccc}
		\hline \hline
		$A$ $[\mu \mathrm{m}]$ & $B$ $[\mu \mathrm{m}^3]$ & $\beta$ $[ \mathrm{m}^{3}\mathrm{kg}^{-1}]$ & $D$ $[\mu \mathrm{m}^2\mathrm{s}^{-1}]$
		\\ \hline 
		$1.4\cdot 10^{-3}$  & $3.5\cdot 10^{-4}$ & $3.5\cdot 10^{2}$ & $50$ \\  \hline \hline
	\end{tabular}
\end{table}
\end{center}

\begin{figure*}[htbp]
\centering
\includegraphics[scale=0.31]{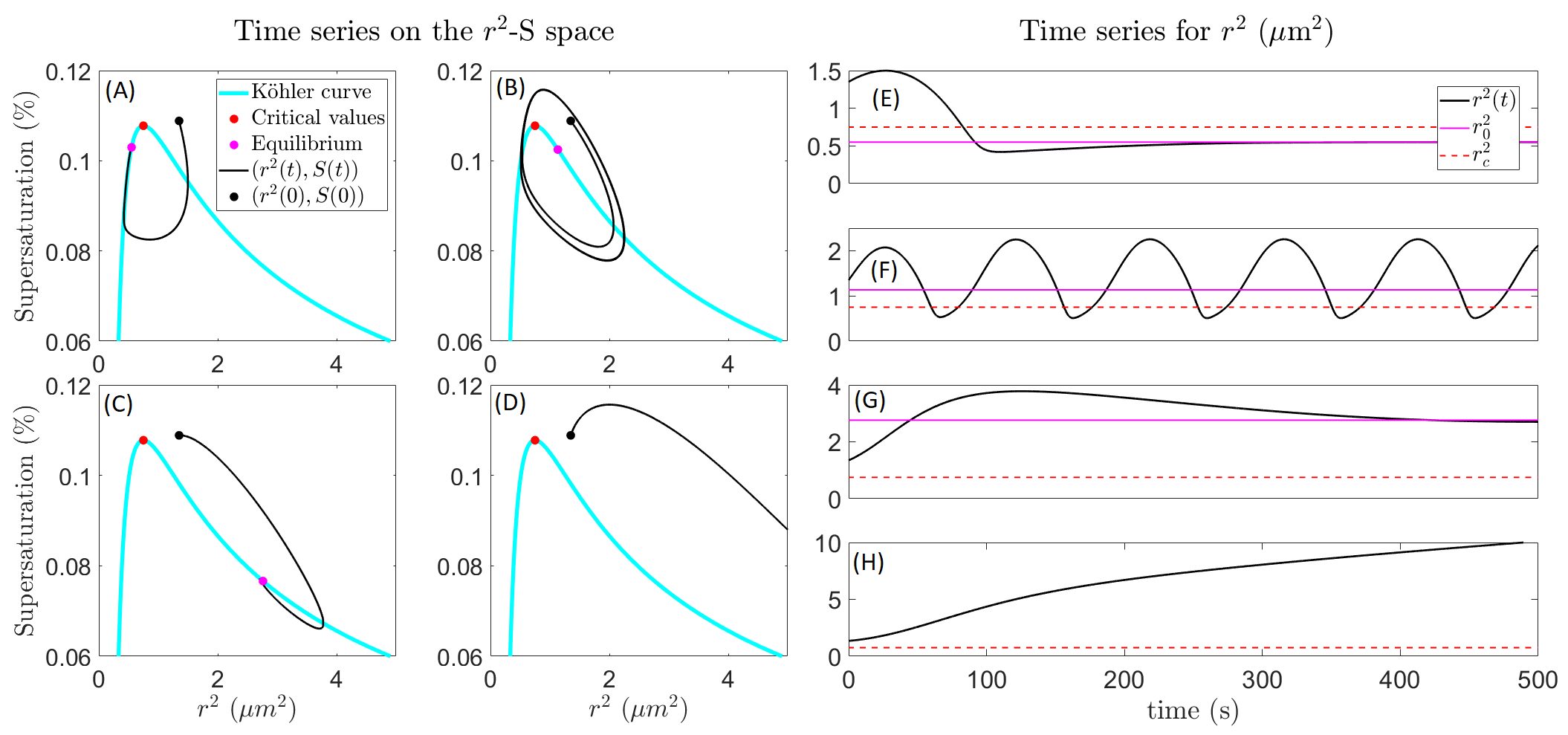}
\caption{\label{fig:ts1}\textbf{Panel (A):} solution of the SRK equation in the $r^2$-$\SSS$ space for the parameters in Table~\ref{tab:1}, $N=50\mathrm{cm}^{-3}$, $\tau^{-1} = 8.4\cdot 10^{-4}\%\mathrm{s}^{-1}$. This gives Jacobian eigenvalues $\lambda = -0.0352 \pm 0.0056 i$. \textbf{Panel (B):} same as panel (A), for $\tau^{-1} = 1.2\cdot 10^{-3}\%\mathrm{s}^{-1}$ and $\lambda = 0.004 \pm 0.168 i$. \textbf{Panel (C):} same as panel (A), for $\tau^{-1} = 1.4\cdot 10^-3\%\mathrm{s}^{-1}$ and $\lambda = 0.004 \pm 0.168 i$. \textbf{Panel (D):} same as panel (A), for $\tau^{-1} = 1.8\cdot 10^{-3}\%\mathrm{s}^{-1}$, where the system has no equilibria. \textbf{Panels (E), (F), (G) and (H):} $r^2(t)$ against time, with the parameter configuration of Panels (A), (B), (C) and (D), respectively. Note that in panels (D) and (H) there are no magenta dot and line, respectively, since there are no equilibria in this case--- see Eq.~\eqref{eq:equilibrium}---.}
\end{figure*}

We now specify droplet number concentration, $N=50\mathrm{cm}^{-3}$ which gives $\alpha = 0.011 (\mu\mathrm{ms})^{-1}$, providing the activation threshold for $\tau^{-1}=A\alpha = 1.54\cdot 10^{-5} \mathrm{s}^{-1}$. The stability of the SRK system is determined by the largest real part of the eigenvalues $\lambda_1,\lambda_2$ of the linearized equation around $(r^2_0,\SSS_0)$ vs. $\tau^{-1}$, here shown in Fig.~\ref{fig:landscape_alpha1} (A). The lines with squares and crosses correspond to the real parts of the eigenvalues, which merge when they become complex, i.e. $\lambda_1 = \overline{\lambda_2}$. The vertical dashed lines indicate the values of $\tau^{-1}$ where $(r^2_0,\SSS_0)$ loses stability and bifurcates. Thus, three distinct regimes are found  and labeled as (R1), (R2) and (R3). At regime (R1), the system is stable and activation cannot occur. At (R2), the equilibrium $(\SSS_0,r^2_0)$ loses stability giving two alternatives: either the system activates leading to $r^2 (t)\rightarrow \infty$ as $t\rightarrow \infty$, or the solutions of the SRK equation converge to a limit cycle \citep{guckenheimer_book}. Finally, at regime (R3), the SRK equation is stable, but converges to an equilibrium radius $r_0$ larger than K\"ohler's critical value $r_c$. To demonstrate these regimes, the SRK equations where numerically integrated and shown in the $r^2$-$\SSS$ space in Fig.~\eqref{fig:ts1}. Panel (A) shows the stable deactivated regime (R1), with $\tau^{-1}<2A\alpha/3$. Panel (B) reveals that the linearly unstable regime (R2) with $\tau^{-1}<A\alpha$ is in fact a limit cycle, where CCN activate and deactivate indefinitely. In panel (C), the system equilibrates at $r^2_0>r^2_c$, i.e. regime (R3). And in panel (D), $\tau^{-1}>A\alpha$ so that CCN activate and, $r^2(t)\rightarrow \infty$. Panels (E), (F), (G) and (H) are the time series of the squared-radius $r^2(t)$ vs. time.

Changes in the particle concentration $N$ affect the stability of $(r^2_0,\SSS_0)$. In Fig.~\ref{fig:landscape_alpha1} (B) we plot the largest real part of the eigenvalues $\lambda_1,\lambda_2$ of $J$ against $\tau^{-1}$ and $N$. Values in the gray triangle, below the black diagonal line lead to an activated family of CCN. Such line is computed from the activation threshold $\tau^{-1} = A\alpha$. The green area indicates configurations that lead to stable and inactive CCN populations. The magenta diagonal line satisfies $\tau^{-1} = 2A\alpha/3$, the values where $r^2_0=r^2_c$. The yellow diagonal strip corresponds to linearly unstable $(r^2_0,\SSS_0)$, which suggest the presence of oscillations as shown in panel (B) of Fig.~\ref{fig:ts1}. 

Regarding the oscillatory behavior shown in Fig.~\ref{fig:ts1}(B), we demonstrate its dependence on $\alpha(N)$ in Fig.~\ref{fig:limit_cycles}. First, we plot the equilibrium $(r^2_0,\SSS_0)$ against $\alpha(N)$, around which the limit cycle oscillates. Panels (B), (C) and (D) show the dependence of the limit cycle amplitude on $N$. The supersaturation rate $\tau^{-1}$ is taken as $5\%$ above the bifurcation point where the SRK system at the equilibrium loses stability. As $N$ grows, the cycle amplitude shrinks and, in fact, for $\alpha(N) >\frac{4A^3}{243B^2}D$, oscillations can no longer exist, as proved in Appendix~\ref{app:proof_2}. Also, the length of the interval of $\tau^{-1}$ sustaining oscillations depends sensitively on the CCN dry radius as shown in Appendix~\ref{eq:b_param}.

\begin{figure}[ht]
\centering
\includegraphics[width=1\linewidth, height=.5\textwidth]{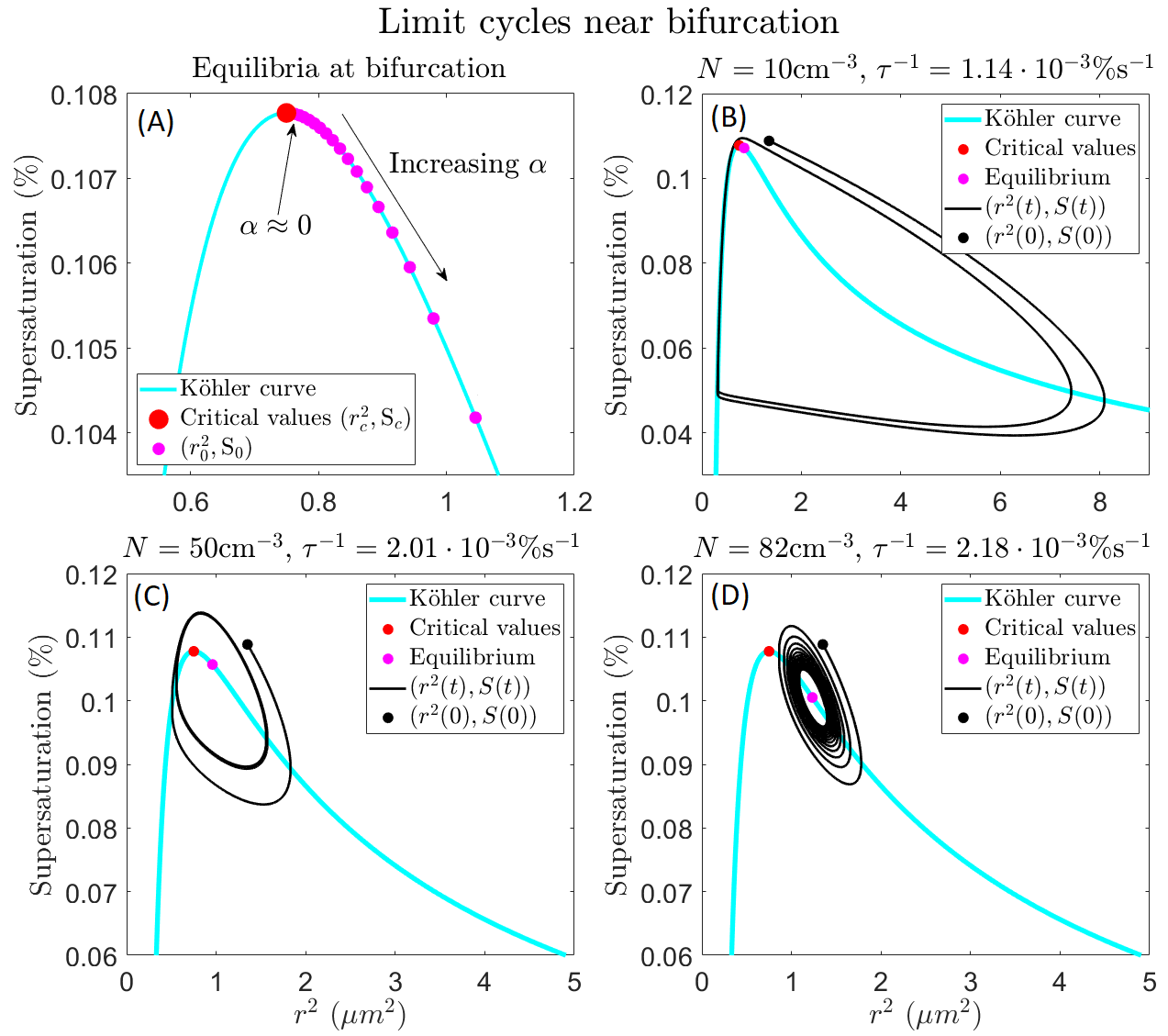}
\caption{\label{fig:limit_cycles}Equilibria and solutions of the SRK equation for the parameters of Table~\ref{tab:1} and for $N$ and $\tau^{-1}$ specified in the titles. \textbf{Panel (A):} Location of the equilibrium $(r^2_0,\SSS_0)$ as a function of $\alpha(N)$. \textbf{Panels (B), (C) and (D):} The black curve shows the solution of the SRK equation starting at the initial conditions shown in the black dot, the cyan line shows the K\"ohler curve, the red dot is the K\"ohler critical values $(r^2_c,\SSS_c)$, and the magenta dot shows the equilibrium given by Eq.~\eqref{eq:equilibrium}. The value of $\tau^{-1}$ is chosen $5\%$ above the point where the equilibrium loses stability, i.e. the bifurcation point.}
\end{figure}

\section{Concluding remarks}\label{sec:discussion}

In this study, we approached the problem of droplet activation from a multi-particle perspective. We coupled the supersaturation budget to K\"ohler's equation and explored the system's behavior as a function of updrafts as the supersaturation source $\tau^{-1}$ and the droplet concentration $N$. Four main conclusions can be drawn from this theoretical analysis: 
\begin{enumerate}
\item[(i)] Only when the supersaturation timescale satisfies $\tau^{-1}>A\alpha$ do CCN activate. For instance, in the conditions of Section~\ref{sec:case_study}, a polluted cloud with concentration around $1000\mathrm{cm}^{3}$ and undergoing updrafts of less than $0.5 \mathrm{m}\mathrm{s}^{-1}$ will not activate.
\item[(i)] Weak updrafts are not enough to attain activation: only when the supersaturation timescale satisfies $\tau^{-1}>A\alpha$ do CCN activate. 
\item[(ii)] When $\tau^{-1}<2A\alpha/3$, the SRK equations are linearly stable, and CCN do not activate, even if the initial sizes of humidified aerosols exceed K\"ohler's critical radius $r_c$.
 \item[(iii)] Droplet radii can stabilize around an equilibrium value $r_0$ greater than K\"ohler's critical radius $r_c$.  When $\tau^{-1} = 2A\alpha/3$, $r_0=r_c$. 
\item[(iv)] In a weakly lifted and sparsely populated air parcel, the SRK equation allows the existence of limit cycles, where CCN activates and deactivates indefinitely. As $\tau^{-1}$ is increased/decreased, the limit cycle appears/disappears, suggesting that the system undergoes a Hopf bifurcation \cite[Chapter~3]{guckenheimer_book}.
\end{enumerate}

Typically, haze growth is not accounted for in cloud models due to the high temporal resolution needed \citep[\S 5.3]{khain_pinsky_2018}, although there are models that determine wet aerosol growth during activation \citep{pinsky_2008,magaritz-ronen_2016a}. Generally, supersaturation is calculated at each grid point and timestep, and so if the critical radius $r_c$, according to which haze particles with a radius larger than $r_c$ are deemed activated. The present analyses, however, show that seemingly active droplets--- with a radius larger than $r_c$--- can remain in equilibrium in weak updraft conditions that imply small supersaturation source, driven by the balance between condensation and supersaturation input, i.e., creating an intermediate regime in between haze and fully activated droplets. Such thermodynamic conditions are especially relevant to small, warm clouds that exhibit weak updrafts.  It is then anticipated that turbulent fluctuations are critical under these conditions as they could trigger transitions between haze and activated droplets \citep{prabhakaran2020}.      

The extension of the analytical investigation of K\"ohler's equation coupled to supersaturation should be done in two main directions. One concerns the effects of turbulent supersaturation fluctuations and their role in modifying activation thresholds, as is observed in general hysteretic processes \citep{berglund_2002}. It is, then, expected that the addition of noise--- coming from turbulent updrafts, for instance--- will modify firstly, K\"ohler's critical radius and supersaturation and, secondly, the stability of the SRK equation. A second direction should address lifting the monodisperse assumption and analyzing its effects on activation following the work of \cite{pinsky2014}, where the supersaturation budget and size distribution evolution are investigated. The addition of CCN with different dry radii or chemical compositions will translate into a substantial increase in the problem's dimensionality, making it more difficult to extract analytical information. Suitable averaging, mean-field, or dimension-reduction techniques could help in this task.

\codeavailability{The used code is available upon request.} 

\dataavailability{No data was used or produced in this study.} 




\appendix


\section{Thermodynamic parameters and notations}\label{app:notation}
There are key thermodynamic parameters employed throughout the paper. All of these parameters are also found in the literature \citep{rogers1989,prupp,khain_pinsky_2018}. We provide the explicit expressions here:
\begin{equation}
	\beta = \frac{1}{q_v} + \frac{L_w^2}{c_pR_vT^2}
\end{equation}
where $q_v [-]$ is the water vapor mixing ration, $L_w[\mathrm{J}\mathrm{k}^{-1}]$ is the latent heat of evaporation, $c_p[\mathrm{J}\mathrm{kg}^{-1}\mathrm{K}^{-1}]$ is the specific heat capacity of moist air at constant pressure, $R_v[\mathrm{J}\mathrm{kg}^{-1}\mathrm{K}^{-1}]$ is the specific gas constant for water vapor and $T [\mathrm{K}]$ is the temperature. 

The adiabatic parameter $a_0$ is defined as:
\begin{equation}
    a_0 = \frac{g}{T}\left( \frac{L_w}{c_pR_vT} - \frac{1}{R_a} \right),
\end{equation}
where $g [\mathrm{m}\mathrm{s}^{-2}]$ is the acceleration due to gravity and $R_a [\mathrm{J}\mathrm{kg}^{-1}\mathrm{K}^{-1}]$ is the specific gas constant of moist air.

The diffusional parameter is:
\begin{equation}
	D = \left( \frac{\rho_wL_w^2}{kR_vT^2} + \frac{\rho_wR_vT}{E_w(T)D_{eff}} \right)^{-1},
\end{equation}
where $\rho_w [\mathrm{kg}\mathrm{m}^{-3}]$ is the density of liquid water, $k[\mathrm{Jm}^{-1}\mathrm{s}^{-1}\mathrm{K}^{-1}]$ is the air heat conductivity, $E_w(T)[\mathrm{N}\mathrm{m}^{-2}]$ is the saturation vapor pressure over liquid water and $D_{eff}[\mathrm{m}^2\mathrm{s}^{-1}]$ is the water vapor diffusion in air coefficient.

Regarding the K\"ohler parameters, these are:
\begin{equation}
	A = \frac{2\sigma_w}{\rho_wR_vT},
\end{equation}
where $\sigma_w[\mathrm{N}\mathrm{m}^{-1}]$ is the surface tension equal to the work needed to increase the surface by a unit of square. Raoult's coefficient is:
\begin{equation}\label{eq:b_param}
	B = r_d^3\frac{\nu_N\Phi_s\delta_sM_w\rho_N}{M_N\rho_w},
\end{equation}
where $r_d[\mathrm{m}]$ is the dry radius of the aerosol, $\nu_N[-]$ is the total number of ions produced by salt, $\Phi_s[-]$ is the molecular osmotic coefficient of a deviation from perfect solutions, $\delta_s[-]$ is the soluble fraction of the aerosol, $M_w[\mathrm{Da}]$ is the molecular mass of water, $M_N[\mathrm{Da}]$ is the molecular mass of salt and $\rho_N[\mathrm{kg}\mathrm{m}^{-3}]$ is the density of salty aerosols.

\section{Stability and deactivation interval}\label{ap:proof}

The linear stability is determined by the eigenvalues of the Jacobian of the SRK equation evaluated at the equilibria $(\SSS_0,r^2_0)$ of Eq.~\eqref{eq:equilibrium}, which reads as:
\begin{equation}\label{eq:jacobian}
 \begin{aligned}
J(\SSS_0,r^2_0)&= \begin{bmatrix}	
		-\alpha r_0 & - \frac{\alpha}{2}\SSS_0 r_0^{-1} \\
		 			2D & 3DBr_0^{-5}-DAr_0^{-3}	
\end{bmatrix}\\
&=\begin{bmatrix}	
		-\alpha r_0 & - \frac{1}{2}  \tau^{-1} r_0^{-2}  \\
		 			2D & Dr_0^{-3}(\frac{3}{\tau \alpha} - 2A)
\end{bmatrix}
 \end{aligned}
\end{equation}
In the limit of $\tau^{-1} \longrightarrow 0$, we have that $J(\SSS_0,\RRR_0)$ is lower-triangular with negative diagonal entries yielding the stability of the equilibrium. For $\tau^{-1}>0$,  it is challenging to find an analytical expression for the eigenvalues as a function of $\tau^{-1}$, yet there is a non-empty interval for which stability holds. This is gathered in the following proposition:
\begin{proposition}\label{prop:1}
	Consider the SRK equation, and its Jacobian--- Eq.~\eqref{eq:jacobian}--- evaluated at the equilibrium $(\SSS_0,r^2_0)$ given in Eq.~\eqref{eq:equilibrium}. Suppose that the coefficients $\alpha,D,B,A$ are positive. Then for all $\tau^{-1} $ in the interval $(-\infty , 2A\alpha/3]$, the equilibrium $(\SSS_0,r^2_0)$ is linearly stable.
\end{proposition}
\begin{proof}
	To prove the stability of $(\SSS_0,r^2_0)$ we have to examine the real parts of the eigenvalues $\lambda_1$ ,and $\lambda_2$ of the Jacobian matrix  $J$ given in \eqref{eq:jacobian}. 
The trace  $\mathrm{Tr}(J)$ and determinant $|J|$, provide the well-known analytic expressions
	\begin{subequations}\label{eq:eigenvalues}
		\begin{align}
			\lambda_1 &= \frac{\mathrm{Tr}(J) + \sqrt{\mathrm{Tr}(J)^2 - 4|J|}}{2},\\
			\lambda_2 &= \frac{\mathrm{Tr}(J) - \sqrt{\mathrm{Tr}(J)^2 - 4|J|}}{2}.
		\end{align}
	\end{subequations}
 Simple calculations show that	
\begin{equation}
 \begin{aligned}
		|J| &= -\alpha D r^{-2}_0 \left(\frac{3}{\tau \alpha} - 2A \right) + \tau^{-1} D r^{-2}_0
		\\&= 2Dr^{-2}_0 \left( A\alpha - \tau^{-1} \right). 
 \end{aligned}
\end{equation}
ensuring positivity of $|J|$ if and only if $\tau^{-1} < A\alpha$. Thus, from Eq.~\eqref{eq:eigenvalues}, the sign of the trace $\mathrm{Tr}(J)$ will determine the stability in the interval $\tau^{-1} < A\alpha$ when $|J|\geq 0$ . Indeed, if the discriminant $\mathrm{Tr}(J)^2 - 4|J|$ is negative, its square-root is a pure-imaginary value and, thus, the stability of $(\SSS_0,r^2_0)$ is determined by the sign of $\mathrm{Tr}(J)$. If the discriminant is positive, we have the following triangle inequality:
 \begin{equation}\label{eq:triangle}
     \left|\mathrm{Tr}(J)\right|\geq \sqrt{\mathrm{Tr}(J)^2 - 4|J|},
 \end{equation}
 and hence, the trace $\mathrm{Tr}(J)$ decides the stability. To see this, if $\mathrm{Tr}(J)$ is negative, $\lambda_2$ in Eq.~\eqref{eq:eigenvalues} will clearly remain negative, but also $\lambda_1$  by virtue of Eq.~\eqref{eq:triangle}. Conversely, if $\mathrm{Tr}(J)$ is positive, $\lambda_1$ will remain so and that is enough to ensure instability.

	The trace is written as:
	\begin{equation}
		\mathrm{Tr}(J) = -\alpha r_0+D\alpha^{-1}r_0^{-3}\left(3\tau^{-1} - 2A\alpha \right),
	\end{equation}
	from where we deduce that for $\tau^{-1}<2A\alpha/3$, the trace is negative and, hence, $2A\alpha/3$ is a lower bound to the smallest real zero of $\mathrm{Tr}(J)$. This together with the determinant being positive for $\tau^{-1}<A\alpha$, we conclude that the equilibrium $(\SSS_0,r_0^2)$ is stable for $\tau^{-1}<2A\alpha/3$.
\end{proof}

\subsection{Oscillatory regime}\label{app:proof_2}
For the occurrence of limit cycles, it is necessary that the Jacobian eigenvalues develop an imaginary part, so that $\lambda_1 = \bar{\lambda_2}$ \citep{guckenheimer_book}. The value of $\tau^{-1}$ at which this happens is given by the equation $\mathrm{Tr}(J(\SSS_0,r^2_0))^2 = 4|J(\SSS_0,r^2_0)|$ which is equivalent to solving the following polynomial equation:
\begin{equation}
 \begin{aligned}
\label{eq:polynomial_imaginary}
  &  -\alpha^2B^4 + D^2\alpha^{-2}(3\tau^{-1}-2A\alpha)(A-\tau^{-1}\alpha^{-1})^4 \\
   &- 2DB^2(3\tau^{-1}-2A\alpha)(A-\tau^{-1}\alpha^{-1})^2 \\
   &- 8DB^2(A\alpha-\tau^{-1})(A-\tau^{-1}\alpha^{-1})^2 = 0.
 \end{aligned}
\end{equation}
The solutions of this equation in the variable $\tau^{-1}$, would give the location of the intersection of the two branches in Fig.~\ref{fig:landscape_alpha1}. When eigenvalues have negative real parts, their imaginary component correspond to the oscillations in a damped equilibrium. Positivity of the real part yields instability and, the possible appearance of limit cycles. Notice that when $B=0$, a solution of Eq.~\eqref{eq:polynomial_imaginary} is given by $\tau^{-1}=A\alpha$.

The range of values of $\tau^{-1}<A\alpha$ where the equilibrium $(r^2_0,\SSS_0)$ is unstable is susceptible of supporting a limit cycle, as illustrated in Fig.~\ref{fig:ts1}B. For that, it is necessary the trace of the Jacobian $\mathrm{Tr}(J)$ is positive, as per Eq.~\eqref{eq:eigenvalues}. The zeros of the trace $\mathrm{Tr}(J)$ satisfy the following cubic polynomial equation in the indeterminate $\tau^{-1}$:

\begin{equation}\label{polynomial}
 \begin{aligned}
		P(\tau^{-1}) = 3D&\tau^{-3} - 8DA\alpha\tau^{-2} \\
		&+ 7DA^2\alpha^2\tau^{-1} - B^2\alpha^4 - 2DA^3\alpha^3 = 0.
 \end{aligned}
\end{equation}
The roots of this polynomial are given explicitly in Appendix~\ref{app:roots}. However, an examination of the signs of the coefficients reveals--- by Descartes' rule of signs--- that $P$ has either one or three positive roots. Conversely, the reverse-sign polynomial $P_-(\tau^{-1}) $ reads as
\begin{equation}
 \begin{aligned}
		P_-(\tau^{-1}) &= P (-\tau^{-1}) \\
		&= -3D\tau^{-3} - 8DA\alpha\tau^{-2}\\
		&\hspace{8ex} - 7DA^2\alpha^2\tau^{-1} - B^2\alpha^4 - 2DA^3\alpha^3 .
 \end{aligned}
\end{equation}
The lack of sign-change in the coefficients of $P_-(\tau^{-1})$ shows that $P(\tau^{-1})$ does not have zeros for negative values of $\tau^{-1}$. As a consequence, it is possible that $(\SSS_0,r^2_0)$ loses its stability for some values of $\tau^{-1}$ prior to the activation threshold, namely in the range $2A\alpha/3<\tau^{-1}<A\alpha$. It is challenging, however, to determine whether in that range droplets activate yielding that the solutions $r^2(t)$ diverge to infinity, or whether they converge to a limit cycle, see e.g. \cite{guckenheimer_book}. Below we provide a sufficient condition on $\alpha(N)$ for the existence of positive eigenvalues for the trace $\mathrm{Tr}(J)$ for values of $\tau^{-1} < A\alpha$.
\begin{proposition}\label{prop:2}
	In the conditions of Proposition~\ref{prop:1}, if $\alpha < \alpha_{max}=4A^3D/(243B^2)$, there exists an interval  $(\tau^{-1}_{a},\tau^{-1}_{b}) \subset (0,A\alpha)$ such that the trace of the Jacobian $\mathrm{Tr}(J)$ is positive.
\end{proposition}

\begin{proof}
	Since $\mathrm{Tr}(J)$ is a function of $\tau^{-1}$, we shall define, for notational convenience, $T(\tau^{-1}) = \mathrm{Tr}(J)$, for every $\tau^{-1}$ in $(0,A\alpha)$.
	Additionally, we note that $T$ is a smooth function in the interval $(0,A\alpha)$ and so is the polynomial $P$.
	
	\textbf{Step 1.} As also noted in Eq.~\eqref{polynomial}, the zeros of the trace $\mathrm{Tr}(J)$ are equivalent to those of $P(\tau^{-1})$ in the variable $\tau^{-1}$. This is to say that $P(\tau^{-1}) = 0$ if an only if $T(\tau^{-1})=0$. However, this is only valid in the interval $(0,A\alpha)$, since the following limit holds:
	\begin{equation}
		\lim_{\tau^{-1}\rightarrow A\alpha^{-}} T(\tau^{-1}) = -\infty.
	\end{equation}
	Furthermore, we have that $T(\tau^{-1})>0$ if and only if $P(\tau^{-1})>0$ and $T(\tau^{-1})<0$ if and only if $P(\tau^{-1})<0$.
	
	\textbf{Step 2.} The local maximum and local minimum of $P$ are located at $\tau^{-1} = 7A\alpha/9$ and $A\alpha$, respectively. This is calculated from the derivative of $P(\tau^{-1})$:
	\begin{equation}\label{eq:der_pol}
		P'(\tau^{-1}) = 9D\tau^{-2} - 16DA\alpha\tau^{-1} + 7DA^2\alpha^2 = 0.
	\end{equation}
	From where we obtain the roots $\tau^{-1} = 7A\alpha/9$ and $A\alpha$. Since we are only interested in the interval $(0,A\alpha)$, we shall focus on the value $\tau^{-1} = 7A\alpha/9$. Because cubic term coefficient $3D$ is positive, it follows that $\tau^{-1} = 7A\alpha/9$ is a local maximum.
	
	\textbf{Step 3.} We turn to examine the sign of $P(7A\alpha/9)$. A direct evaluation provides the following expression:
	\begin{equation}
		P\left(\frac{7A\alpha}{9}\right) = \frac{4}{243}DA^3\alpha^3 - B^2\alpha^4.
	\end{equation}
	Equalizing to zero, gives that if $\alpha < \frac{4A^3}{243B^2}D$, $P(7A\alpha/9)>0$. Hence, in such case, there exists $\varepsilon > 0$ such that $P(\tau^{-1})>0$ for every $\tau^{-1}$ in the interval $(7A\alpha/9 - \varepsilon,7A\alpha/9 + \varepsilon)$. 
	
	It is now enough to define $\tau^{-1}_{a}=7A\alpha/9 - \varepsilon$ and $\tau^{-1}_{b}=7A\alpha/9 + \varepsilon$ to obtain the desired result.
	
\end{proof}

\section{Dependence of oscillatory interval on $B$}

While Raoult's factor $B$, defined on Eq.~\eqref{eq:b_param} of Appendix~\ref{app:notation}, does not play a role in determining the activation threshold, it is key in the presence of positive eigenvalues of the linearized SRK equation, as per Proposition~\ref{prop:2}. It was there found a maximum value $\alpha_{max}$ such that for every $\alpha(N)$ lower than $\alpha_{max}$, the equilibrium $(\SSS_0,r^2_0)$ in Eq.~\eqref{eq:equilibrium} loses stability for supersaturation rates $\tau^{-1}$ lower than the activation threshold $A\alpha(N)$. Such maximum value $\alpha_{max}$ is:
\begin{equation}
	\alpha_{max} = \frac{4A^3}{243B^2}D.
\end{equation}
In Fig.~\ref{fig:b_dependence} we plot the values of $\alpha_{max}$ as a function of aerosol dry radius, for salt (blue line) and ammonium sulfate (red line). We observe the sensitive dependence of $\alpha_{max}$ on dry radius. As a consequence of this, we can conclude that the interval of values of $\tau^{-1}$ that can lead to oscillatory behavior--- analogous to that illustrated in Fig.~\ref{fig:ts1}(B)--- increases with smaller dry aerosol radius.
\begin{figure}[t]
	\centering
	\includegraphics[width=.9\linewidth, height=.33\textwidth]{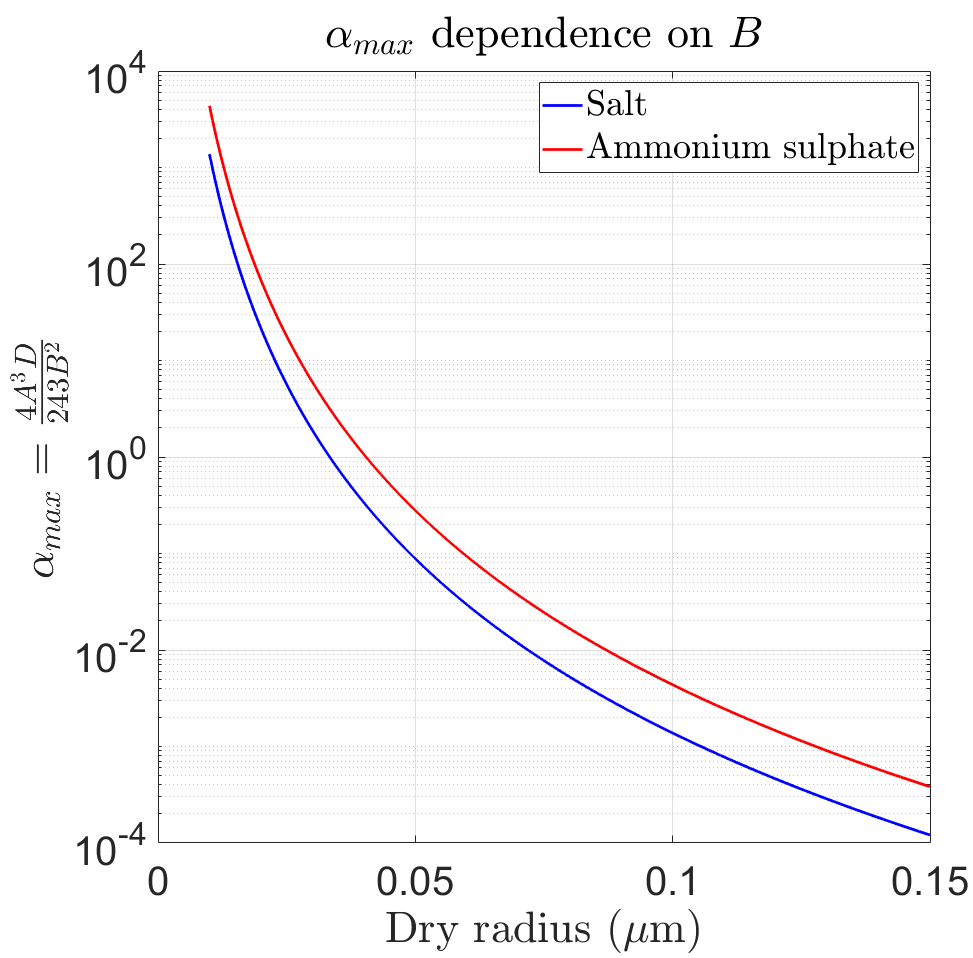}
	\caption{\label{fig:b_dependence}We plot $\alpha_{max}=\frac{4A^3}{243B^2}D$ as a function of aerosol dry radius and for two inorganic compounds indicated in the legend.}
\end{figure}

\section{Roots of the polynomial $P$ in Eq.~\eqref{polynomial}}\label{app:roots}

The roots $\{x_1,x_2,x_3 \}$ of the cubic polynomial $P(\cdot)$ with indeterminate $\tau^{-1}$ introduced in Eq.~\eqref{polynomial} are given by:

\begin{equation}\label{eq:smallest_root}
	x_1 = \sqrt[3]{\xi_1 + \sqrt{\xi_1^2 + \xi_2^3}} +\sqrt[3]{\xi_1 - \sqrt{\xi_1^2 + \xi_2^3}}+\frac{8A\alpha}{9},
\end{equation}

\begin{equation}
 \begin{aligned}
	x_2 &= \left(-\frac{1}{2}+i\frac{\sqrt{3}}{2}\right)\sqrt[3]{\xi_1 + \sqrt{\xi_1^2 + \xi_2^3}} \\
	&\hspace{5ex}+\left(-\frac{1}{2}-i\frac{\sqrt{3}}{2}\right)\sqrt[3]{\xi_1 - \sqrt{\xi_1^2 + \xi_2^3}}+\frac{8A\alpha}{9},
 \end{aligned}
\end{equation}
\begin{equation}
 \begin{aligned}
	x_3 &= \left(-\frac{1}{2}-i\frac{\sqrt{3}}{2}\right)\sqrt[3]{\xi_1 + \sqrt{\xi_1^2 + \xi_2^3}} \\
	&\hspace{5ex}+\left(-\frac{1}{2}+i\frac{\sqrt{3}}{2}\right)\sqrt[3]{\xi_1 - \sqrt{\xi_1^2 + \xi_2^3}}+\frac{8A\alpha}{9},
 \end{aligned}
\end{equation}
where $\xi_1$ and $\xi_2$ are defined as:
\begin{align}
	\xi_1 &= -\frac{1}{3^6}(A\alpha)^3 - \frac{B^2\alpha^4}{6D},\\
	\xi_2&= -\frac{1}{3^4}(A\alpha)^2.
\end{align}

If $\alpha < 4A^3D/(243B^2)$, the polynomial $P$ has three real roots, as per Proposition~\ref{prop:2}. In this case, we assume that $x_1\leq x_2 \leq x_3$. In particular, if $B=0$, then $x_1 = 2A\alpha/3$, as  observed by direct evaluation of $P(2A\alpha/3)$. Moreover, $x_2=x_3 = A\alpha$ in this case. With this remark and the location of the local maximum of $P$ at $7A\alpha/9$ obtained by solving Eq.~\eqref{eq:der_pol}, we can estimate $x_1$ as the midpoint between $2A\alpha/3$ and $7A\alpha/9$:

\begin{equation}
	x_1 \approx \frac{1}{2}\left( \frac{2}{3}+\frac{7}{9}\right)A\alpha=\frac{13}{18}A\alpha.
\end{equation}

\noappendix       

\appendixfigures  

\appendixtables   


\authorcontribution{MSG and MDC conceived the present idea. MSG lead the analyses, MDC and IK supported. MSG, MDC and IK discussed the results and wrote the manuscript. All authors critically contributed to the final form of the manuscript.} 

\competinginterests{The authors declare no competing interests.} 


\begin{acknowledgements}
This work has been partially supported by the European Research Council (ERC) under the European Union’s Horizon 2020 research and innovation program (grant agreement no. 810370). MSG is grateful to the Feinberg Graduate School for their support through the Dean of Faculty Fellowship.
\end{acknowledgements}



\bibliography{apssamp}

\end{document}